\documentclass[letterpaper, 10 pt, conference]{ieeeconf}  

\IEEEoverridecommandlockouts                              

\usepackage{graphicx}        
\usepackage{amssymb}
\usepackage{graphicx}
\usepackage{algorithm}
\usepackage{algpseudocode}
\usepackage{amsmath,todonotes}
\usepackage{mathtools}
\usepackage{comment}
\usepackage{multirow}
\usepackage{booktabs}
\usepackage{cite}
\usepackage{balance}

\newtheorem{proposition}{Proposition}
\newtheorem{lemma}{Lemma}

\newtheorem{theorem}{Theorem}

\newtheorem{definition}{Definition}

\title{\LARGE \bf
A Network Monitoring Game with Heterogeneous \\ Component Criticality Levels
}
\author{Jezdimir Milo\v{s}evi\'{c}$^1$, Mathieu Dahan$^2$, Saurabh Amin$^3$, Henrik Sandberg$^1$
\thanks{$^1$The Division of Decision and Control Systems, School of Electrical Engineering and Computer Science, KTH Royal Institute of Technology, Stockholm, Sweden. Emails: \{jezdimir, hsan\}@kth.se; 
$^2$Center for Computational Engineering, Massachusetts Institute of Technology, Cambridge, Massachusetts, Email: mdahan@mit.edu
$^3$Department of Civil and Environmental Engineering, and Institute for Data, Systems, and Society, Massachusetts Institute of Technology, Cambridge. Email: amins@mit.edu. 
}
}

\begin{document}

\maketitle
%
\begin{abstract} 
We consider an attacker-operator game for monitoring a large-scale network that is comprised on
components that differ in their criticality levels. 
In this zero-sum game, the operator seeks to position a limited number of sensors to monitor the network against an attacker who strategically targets a network component. 
The operator (resp. attacker) seeks to minimize (resp. maximize) the network loss. 
To study the properties of mixed-strategy Nash Equilibria of this game, we first study two simple instances: 
(i)~When component sets monitored by individual sensor locations are mutually disjoint;
(ii)~When only a single sensor is positioned, but with possibly overlapping monitoring component sets. 
Our analysis reveals new insights on how criticality levels impact the players’ equilibrium strategies. 
Next, we extend a previously known approach to obtain an approximate Nash equilibrium for the general case of the game. 
This approach uses solutions to minimum set cover and maximum set packing problems to construct an approximate Nash equilibrium. 
Finally, we implement a column generation procedure to improve this solution and numerically evaluate the performance of our approach.
\end{abstract}

\section{Introduction}

Critical infrastructure networks such as water distribution or power networks are attractive targets for malicious attackers~\cite{sandberg2015cyberphysical,weerakkody2019challenges}. 
In fact, successful attacks against these networks have already been documented~\cite{slay2007lessons,case2016analysis}, amplifying the need for development of effective defense strategies. 
An important part of a defense strategy is attack detection~\cite{nist}, which can be achieved  
by deployment of sensors to monitor the network~\cite{dan2010stealth,2017arXiv170500349D,krause2011randomized}. 
However, if a network is large, it is expected that the number of sensors would be insufficient to enable monitoring of the entire network. 
Hence, the problem that naturally arises is how to strategically allocate a limited number of sensors in that case. 

We adopt a game theoretic approach to tackle this problem. 
So far, game theory has been used for studying various security related problems~\cite{zhu2015game,MIAO201855,7498672,shreyasinvestment,pita2008deployed,washburn1995two, bertsimas2016power}, including the ones on sensor allocation. 
The existing works considered developing both static~\cite{stack_metju,pirani2018game,ren2018secure} and randomized (mixed) monitoring strategies~\cite{2017arXiv170500349D,krause2011randomized}.
Our focus is on randomized strategies, which are recognized to be more effective than static once the number of sensors to deploy is limited~\cite{krause2011randomized,2017arXiv170500349D}.

Our game model is related to the one in~\cite{2017arXiv170500349D}.
The network consists of the components to be monitored, and 
sensor locations can be selected from the predefined set of nodes.
From each node, attacks against a subset of components can be detected.  
However, while~\cite{2017arXiv170500349D} studies the game where the players (the operator and the attacker) make decisions based on so-called detection rate, in our game the decisions are made based on the component criticality. 
This game model is motivated by the risk management process, where one first conducts a risk assessment to identify the critical components in the system, and then allocates resources based on the output of the assessment~\cite{nist}. 
Particularly, the operator seeks placing a limited number of sensors to minimize the loss that is defined through the component criticality, while the attacker seeks to attack a component to maximize it.

A monitoring strategy we aim to find is the one that lies in a Nash Equilibrium (NE) of the game. 
Since our game is a zero-sum game, a NE can be calculated by solving a pair of linear programs~\cite{basar1999dynamic}. 
However, these programs are challenging to solve in our case, since the number of actions of the operator grows rapidly with the number of sensors she has at disposal.
Moreover, a NE calculated using this numerical procedure usually does not provide us with much intuition behind the players' equilibrium strategies. 
Our objective in this work is to: 
(i)~Study how the components' criticality influences the equilibrium strategies of the players; 
(ii)~Investigate if some of the tools from~\cite{2017arXiv170500349D} can be used to calculate or approximate an equilibrium monitoring strategy for our game in a tractable manner.

%

Our contributions are threefold. 
Firstly, for a game instance where component sets monitored by individual sensor locations are mutually disjoint, we characterize a NE analytically (Theorem~\ref{theorem:analytical_solution_special_case}). 
This result provide us with valuable intuition behind the equilibrium strategies, and reveals some fundamental differences compared to the game from~\cite{2017arXiv170500349D}.
Particularly, the result illustrates how the components' criticality influences strategies of the players, that the resource limited operator can leave some of the noncritical components unmonitored, and that the attacker does not necessarily have to attack these components.  
We also consider a game instance where a single sensor is positioned but the monitoring sets are allowed to overlap, and extend some of the conclusions to this case (Proposition~\ref{theorem:solution_special_case_3}). 

%

%
%

Secondly, we show that the mixed strategies proposed in~\cite{2017arXiv170500349D} can be used to obtain an approximate NE.  
In this approximate NE, the monitoring (resp. attack) strategy is formed based on a solution to minimum set cover (resp. maximum set packing) problem. 
A similar approach for characterizing equilibria was also used in~\cite{pita2008deployed,washburn1995two, bertsimas2016power}, yet for specific models and player resources.     
Our analysis reveals that these strategies may represent an exact or a relatively good approximation of a NE if the component criticality levels are homogeneous, while the approximation quality decreases if the gap in between the maximum and the minimum criticality level is large~(Theorem~\ref{thm:mix_strategies_diff_indexes}).  

Finally, we discuss how to improve the set cover monitoring strategy from the above-mentioned approximate equilibrium.  
The first approach exploits the intuition from Theorem~\ref{theorem:analytical_solution_special_case}. 
Particularly, if a group of the components have a criticality level sufficiently larger then the others, we show that the strategy can be improved by a simple modification (Proposition~\ref{thm:binary_weights}). 
The second approach is by using a column generation procedure (CGP)~\cite{desrosiers2005primer}.
This procedure was suggested in~\cite{2017arXiv170500349D} as a possible way to improve the set cover strategy, but it was not tested since the strategy already performed well. 
We show that CGP can be applied in our game as well, and test it on benchmarks of large scale water networks. 
The results show that: 
(i)~The running time of CGP rapidly grows with the number of deployed sensors, but the procedure can still be used for finding an equilibrium monitoring for water networks of several hundred nodes; 
(ii)~Running a limited number of iterations of CGP can considerably improve the set cover monitoring strategy.


%
The paper is organized as follows. 
In Section~\ref{section:security_game}, we introduce the game. 
In Sections~\ref{section:exact}, we discuss two special game instances. 
In Sections~\ref{section:approximate_strategies}, we show that the strategies from~\cite{2017arXiv170500349D} can be used to obtain an approximate NE, and discuss how the monitoring strategy from this approximate equilibrium can be further improved.
In Section~\ref{section:simulations}, we test CGP.
In Section~\ref{section:conclusion}, we conclude.


\section{Game Description} \label{section:security_game}

%
Our network model considers a set of components $\mathcal{E}$$=$$\{e_1,\ldots,e_m\}$ that can be potential targets of an attacker, and a set of nodes $\mathcal{V}$$=$$\{v_1,\ldots,v_n\}$ that can serve as sensor positions for the purpose of monitoring.  
By placing a sensor at node $v$, one can monitor a subset of components $E_v$$ \subseteq$$ \mathcal{E}$, which we refer to as the monitoring set of $v$.
Without loss of generality, we assume $E_v$$ \neq$$ \emptyset$, and that every component can be monitored from at least one node. 
If sensors are positioned at a subset of nodes $V$$\subseteq $$\mathcal{V}$, then the set of monitored components can be written as $E_V$$\coloneqq$$\cup_{v\in V}E_v$. 
We refer the reader to Fig.~\ref{figure:example_0} for an illustration of monitoring~sets. 

     \begin{figure}[t]   
    \centering
  \includegraphics[width=75mm]{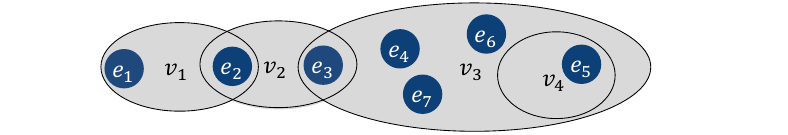}
  \caption{The set of nodes (resp. components) is $\mathcal{V}$$=$$\{v_1,\ldots,v_4\}$ (resp. $\mathcal{E}$$=$$\{e_1,\ldots,e_7\}$). The monitoring sets are $E_{v_1}$$=$$\{e_1,e_2\}$, $E_{v_2}$$=$$\{e_2,e_3\}$, $E_{v_3}$$=$$\{e_3,\ldots,e_7\}$, and $E_{v_4}$$=$$\{e_5\}$. }
  \label{figure:example_0}
\end{figure}

To study the problem of strategic sensor allocation in the network, we adopt a game-theoretic approach. 
Specifically, we consider a zero sum game $\Gamma$$=$$\langle\{1,2\},(\mathcal{A}_1,\mathcal{A}_2),l\rangle$, in which Player~1 (P1) is the operator and Player~2 (P2) is the attacker. 
P1 can select up to $b_1$ nodes from $\mathcal{V}$ to place sensors and monitor some of the network components from $\mathcal{E}$.
We assume that these sensors are protected, in that they are not subject to the actions of P2. 
P2 seeks to select a component from $\mathcal{E}$ to attack.
We assume that if P1 successfully detects the attack, she can start a response mechanism to mitigate the damage.
Thus, in our model, the attack is successful only if it remains undetected by P1. 
Based on the discussion, the action set of P1 (resp. P2) is $\mathcal{A}_1$$=$$\{V$$ \in $$2^\mathcal{V}$$|$$\hspace{1mm}|V|$$\leq$$ b_1 \}$ 
(resp. $\mathcal{A}_2$$=$$\mathcal{E}$). 
The loss function $l$$:$$ \mathcal{A}_1$$ \times$$ \mathcal{A}_2 $$\longrightarrow$$ \mathbb{R}$ is defined by
\begin{equation} \label{eqn:index_and_set_x}
l(V,e):=\begin{cases}
w_{e}, \hspace{2.5mm} e \notin E_V,   \\
\hspace{2mm}0,\hspace{3mm}e \in E_V, 
\end{cases} 
\end{equation}
where $w_{e}$$\in$$ (0,1]$ is a known constant whose value indicates the level of criticality of the component $e$; the assumption $w_e$$>$$0$ is without loss of generality.
For practical purposes, for each $e$$\in$$\mathcal{E}$, $w_e$ can be evaluated as the normalized monetary loss to P1, negative impact on the overall system functionality when the component $e$ is compromised by P2, or a combination of several factors.
We assume that P1 (resp. P2) seeks to minimize (resp. maximize)~$l$.

The players are allowed to use mixed strategies.
A mixed strategy of a player is a probability distribution over the set of her pure actions. 
Particularly, mixed strategies are defined~as
\begin{align*}
& \sigma_{1} \in \Delta_1, \hspace{0.2mm}\Delta_1=\bigg\{\sigma_{1} \in [0,1]^{|\mathcal{A}_1|}\bigg| \sum_{V\in \mathcal{A}_1} \sigma_1(V) =1 \bigg\},\\
& \sigma_{2} \in \Delta_2, \hspace{0.2mm}\Delta_2=\bigg\{\sigma_{2} \in [0,1]^{|\mathcal{A}_2|}\bigg| \sum_{e \in \mathcal{A}_2} \sigma_2(e) =1\bigg\},
\end{align*}  
where $\sigma_{1}$ (resp. $\sigma_{2}$) is a mixed strategy of P1 (resp. P2), and $\sigma_1(V)$ (resp. $\sigma_2(e)$) is the probability the action $V$ (resp. $e$) is taken.  
One interpretation of mixed strategy $\sigma_1$ for P1 is that it provides a randomized sensing plan; similarly for P2. For example, in a day-to-day play in which both players play myopically, P1 (resp. P2) selects sensor placement (resp. attack) plan
according to sampling from probability distribution $\sigma_1$ (resp. $\sigma_2$). 


In the analysis that follows, it is convenient to characterize $\sigma_1$ through the marginal probabilities.  
The marginal probability $\rho_{\sigma_{1}}(v)$ is given by
\begin{align}\label{eqn:marginal}
\rho_{\sigma_{1}}(v) \coloneqq \sum_{ V \in \mathcal{A}_1, v \in V}\sigma_1(V),
\end{align}
and it represents the probability that a sensor is placed at $v$ if P1 plays $\sigma_{1}$. 
Next, given $(\sigma_1$$,$$\sigma_2)$$\in$$ \Delta_1 $$\times $$\Delta_2$, the expected loss is defined by
 \begin{align*}
L(\sigma_1,\sigma_2) \coloneqq \sum_{V \in \mathcal{A}_1} \sum_{e \in \mathcal{A}_2} \sigma_1(V) \sigma_2(e) l(V,e).
\end{align*}
We use $L(V,\sigma_2)$ (resp. $L(\sigma_1,e)$) to denote the case where $\sigma_1(V)$$=$$1$ (resp. $\sigma_2(e)$$=$$1$) for some $V $$\in$$ \mathcal{A}_1$ (resp. $e$$ \in$$ \mathcal{A}_2$). 


We are concerned with strategy profile(s) that represent a NE of $\Gamma$.
A strategy profile $(\sigma^*_1$$,$$\sigma^*_2)$ is a NE if
\begin{align*}
L(\sigma_1^{*},\sigma_2) \leq L(\sigma_1^{*},\sigma_2^{*}) \leq L (\sigma_1,\sigma_2^{*}),
\end{align*} 
holds for all $\sigma_{1}$$ \in$$ \Delta_1$ and  $\sigma_{2}$$ \in$$ \Delta_2$.
We refer to $L (\sigma_1^{*},\sigma_2^{*})$ as the value of the game. 
Thus, given that P2 plays according to $\sigma_2^{*}$, P1 cannot perform better than by playing according to randomized monitoring strategy $\sigma_1^{*}$. 
Additionally, in a zero sum game, the value of the game is equal for every NE. 
Hence, it suffices for P1 to find a single randomized monitoring strategy that lies in equilibrium. 
Similar argument holds for P2's randomized attack strategy $\sigma_2^{*}$.  

We say that a strategy profile $(\sigma^\epsilon_1,\sigma^\epsilon_2)$ is an $\epsilon$--NE of $\Gamma$ if
\begin{align*}
L(\sigma^\epsilon_1,\sigma_2)-\epsilon \leq L (\sigma^\epsilon_1,\sigma^\epsilon_2) \leq L(\sigma_1,\sigma^\epsilon_2)+\epsilon, \hspace{1mm}\epsilon\geq 0,
\end{align*}
for all $\sigma_{1} $$\in $$\Delta_1$ and  $\sigma_{2}$$ \in $$\Delta_2$.
In this case, if P2 plays according to $\sigma^\epsilon_2$, P1 may be able to decrease her loss by deviating from $\sigma^\epsilon_1$, but not more than $\epsilon$. 
Thus, if $\epsilon$ is small enough, $\sigma^\epsilon_{1}$ represents a good suboptimal strategy; similarly for P2. 

Since $\Gamma$ is a zero-sum game with finite number of player actions, equilibrium strategies and the value of the game in a NE exists, and can be obtained by solving the following pair of linear programs~\cite{basar1999dynamic} 
\begin{equation*}\label{eqn:original_LPs}
  \begin{aligned} 
(\text{LP}_1)\hspace{2mm}&\underset{z_1,\sigma_1 \in \Delta_1 }{\text{minimize }} z_1 \hspace{2mm}\text{subject to}\hspace{2mm}L(\sigma_1,e)\leq z_1, \forall e\in\mathcal{A}_2, \\
(\text{LP}_2)\hspace{2mm}&\underset{z_2,\sigma_2\in \Delta_2}{\text{maximize }} z_2\hspace{2mm}
\text{subject to}\hspace{2mm}L(V,\sigma_2)\geq z_2, \forall V\in \mathcal{A}_1.
 \end{aligned}
 \end{equation*}
Yet, these LPs can be computationally challenging to solve using standard optimization solvers for realistic instances of $\Gamma$.
Namely, since the cardinality of $\mathcal{A}_1$ rapidly grows with respect to $b_1$, so does the number of variables (resp. constraints) of $\text{LP}_1$ (resp.  $\text{LP}_2$). 
In the following section, we provide structural properties of equilibria for two simple but instructive cases.
Subsequently, we discuss an approach to compute  $\epsilon$--NE, and then discuss how to further improve the monitoring strategy from this $\epsilon$--NE.


\section{Exact Equilibrium Strategies}\label{section:exact}
In this section, we first study the game instance in which the monitoring sets are mutually disjoint.
We then analyze the game in which  the monitoring sets can overlap with each other, but P1 can only use a single sensor ($b_1$$=$$1$). 

\subsection{Mutually Disjoint Monitoring Sets} \label{section:first_special_case}

We first derive a NE for an instance of $\Gamma$ where the monitoring sets are mutually disjoint, that is, $E_{v_i} $$\cap$$ E_{v_j}$$=$$\emptyset$ holds for any two nodes $v_i $$\neq $$v_j$. 
Let $e^*_i$ denote the component from $E_{v_i}$ with the largest criticality $w_{e^*_i}$. 
One can identify such a component for each of the monitoring sets, and assume without loss of generality $w_{e^*_1}$$\geq $$\ldots $$\geq$$ w_{e^*_n}$. 
For given~$b_1$, we define $Z(b_1)$ as follows:
\begin{equation} \label{eqn:Zset}
Z(b_1)=\bigg\{j\in \{1,\ldots,n\} \bigg|\frac{j-b_1}{\sum_{i=1}^{j}1/w_{e^*_i}}\leq w_{e^*_j}  \bigg\}.
\end{equation}
We argue that this set determines nodes on which P1 places sensors in a NE. 
Particularly, let $p$ be the largest element of $Z(b_1)$, $E_p$$=$$\{e^*_1,\ldots,e^*_p\}$, $S_p$$=$$\sum_{i=1}^{p}1/w_{e^*_i}$, and
$(\sigma^*_1,\sigma^*_2)$ be a strategy profile that satisfies the following conditions: 
\begin{align}
\label{eqn:def_strategy_analytical}
\rho_{\sigma^*_1}(v_j )&=
\begin{cases} 1-\frac{p-b_1}{w_{e^*_j}S_p},\hspace{2mm}j\leq p,  \\
\hspace{13mm} 0,\hspace{2.2mm}j>p,
\end{cases}\\
\label{eqn:att_strategy_analytical}
\sigma^*_2(e)&=
\begin{cases}\frac{1}{w^*_{e} S_p},\hspace{2mm}e \in E_p,  \\
\hspace{5.5mm} 0,\hspace{2mm}\text{otherwise}. 
\end{cases}
\end{align}
Lemma~\ref{lemma:existance_eq_1} establishes existence of $(\sigma^*_1,\sigma^*_2)$. 
In Theorem~\ref{theorem:analytical_solution_special_case}, we show that this strategy profile is a NE.  

\begin{lemma}~\label{lemma:existance_eq_1}
There exists at least one strategy profile $(\sigma^*_1,\sigma^*_2)$ that satisfies~\eqref{eqn:def_strategy_analytical}--\eqref{eqn:att_strategy_analytical}.
\end{lemma}
\begin{proof}
To prove existence of $\sigma^*_1$, we need to prove: (i)~$\rho_{\sigma^*_1}(v )$$\in$$[0,1]$ for any $v$$\in$$\mathcal{V}$;
(ii)~$\sum_{v\in \mathcal{V}}\rho_{\sigma^*_1}(v)$$=$$b_1$. 
If (i) and (ii) are satisfied, then $\sigma^*_1 $$\in$$ \Delta_1$ from Farkas lemma (see Lemma EC.6.~\cite{2017arXiv170500349D}).

We begin by proving (i). 
Note that $b_1$$ \in$$ Z(b_1)$, so $p$$\geq $$b_1$. 
Then $\frac{p-b_1}{w_{e^*_j}S_p}$$\geq$$ 0$, which implies $\rho_{\sigma^*_1}(v )$$\leq $$1$ for any $v$$\in$$\mathcal{V}$.
From $w_{e^*_1}$$\geq $$\ldots $$\geq$$ w_{e^*_p}$ and~\eqref{eqn:Zset}, we have 
$\frac{p-b_1}{w_{e^*_1 S_p}}$$\leq$$\ldots$$ \leq$$ \frac{p-b_1}{w_{e^*_p S_p}}$$\leq$$ 1.$
Hence, $0 $$\leq $$\rho_{\sigma^*_1}(v )$ must hold for any $v$$\in$$\mathcal{V}$. 
Thus, (i) is satisfied. 
In addition, we have
$$ \sum_{v\in \mathcal{V}}\rho_{\sigma^*_1}(v)\stackrel{\eqref{eqn:def_strategy_analytical}}{=}p-\frac{p-b_1}{S_p}\sum_{i=1}^p \frac{1}{w^*_{e_i}} = b_1,$$
so (ii) holds as well. Thus, $\sigma^*_1 $$\in$$ \Delta_1$. 

Next, we show $\sigma^*_{2} $$\in$$ \Delta_2$. Firstly, we have from~\eqref{eqn:att_strategy_analytical} that $0$$ \leq$$ \sigma^*_{2}(e) $$\leq$$ 1$ for any $e$$\in$$\mathcal{E}$. Moreover,
$$ \sum_{e\in\mathcal{E}} \sigma^*_{2}(e)\stackrel{\eqref{eqn:att_strategy_analytical}}{=}\frac{1}{S_p} \sum_{e\in E_p}\frac{1}{w^*_{e}}=1,$$
so we conclude $\sigma^*_{2} $$\in$$ \Delta_2$.
\end{proof}

\begin{theorem}\label{theorem:analytical_solution_special_case}
If $E_{v_i} $$\cap $$E_{v_j}$$=$$\emptyset$ holds for any two nodes $v_i$$\neq $$ v_j$ from $\mathcal{V}$, then any strategy profile $(\sigma^*_1,\sigma^*_2)$ that satisfies~\eqref{eqn:def_strategy_analytical}--\eqref{eqn:att_strategy_analytical} is a NE of $\Gamma$.
\end{theorem}
\begin{proof}
Let $(\sigma^*_{1},\sigma^*_{2})$ be a strategy profile that satisfies~\eqref{eqn:def_strategy_analytical}--\eqref{eqn:att_strategy_analytical}. 
We know from Lemma~\ref{lemma:existance_eq_1} that at least one such a profile exists. 
We first derive an upper bound on the expected loss if P1 plays ${\sigma}^*_{1}$. 
Assume P2 targets component $e$ that belongs to $E_{v_j}$, $j\leq p$.
Then
\begin{equation}\label{eqn:loss_monitored}
\begin{aligned} 
L({\sigma}^*_{1},e) &=  \hspace{-2mm}\sum_{V \in \mathcal{A}_1}\hspace{-2mm}{\sigma}^*_{1}(V) l(V,e) =  
\hspace{-4mm}\sum_{V \in \mathcal{A}_{1},e \notin E_V}\hspace{-4mm} {\sigma}^*_{1}(V)w_e \\
&=w_e  \hspace{-6mm} \sum_{V \in \mathcal{A}_{1},v_j \notin V}  \hspace{-4mm} {\sigma}^*_{1}(V) 
\stackrel{\eqref{eqn:marginal}}{=}w_{e}(1-\rho_{{\sigma}^*_{1}}(v_j))\\
&\stackrel{\eqref{eqn:def_strategy_analytical}}{=}\frac{w_e}{w_{e^*_j}}\frac{p-b_1}{S_p} \stackrel{(*)}{\leq} \frac{p-b_1}{S_p},
\end{aligned}
\end{equation}
where (*) follows from the fact that $w_{e^*_j}$ is the largest criticality among the components from $E_{v_j}$. 
If $p$$=$$n$, we established that $\frac{p-b_1}{S_p}$ is an upper bound on P1's loss. 
If $p$$<$$n$, there exist nodes that are never selected for sensor positioning, so the components from $E_{v_{p+1}}$$,\ldots,$$E_{v_{n}}$ are never monitored. 
From~\eqref{eqn:index_and_set_x}, by targeting an unmonitored component $e_l$, P2 can achieve the payoff 
$w_{e_l}$.
Note that $w_{e_l}$ cannot be larger than $w_{e^*_{p+1}}$, because $w_{e^*_{p+1}}$ is the largest criticality for the monitoring set $E_{v_{p+1}}$, and $w_{e^*_{p+1}}\geq \ldots \geq w_{e^*_{n}}$ holds for the remaining sets $E_{v_{p+2}},\ldots, E_{v_{n}}$.  
Since $p+1$ does not belong to $Z(b_1)$, it follows from~\eqref{eqn:Zset}
\begin{align*}
w_{e^*_{p+1}} < \frac{p+1-b_1}{S_p+1/w_{e^*_{p+1}}} \Longleftrightarrow  w_{e^*_{p+1}}S_p< p-b_1.
\end{align*}
Thus, the loss associated with any unmonitored component $e_l$ is upper bounded by
$L({\sigma}^*_1,e_l) $$\leq $$w_{e^*_{p+1}}$$<$$ \frac{p-b_1}{S_p}.$
From the later observation and~\eqref{eqn:loss_monitored}, we conclude that the loss of P1 cannot be larger than $\frac{p-b_1}{S_p}$. 

Consider now $\sigma_2^*$.
For any $V$, such that $|V|\leq b_1$, we have
\begin{align*}
L(V,\sigma^*_2) &=  \sum_{e \in \mathcal{A}_2} \sigma^*_2(e) l(V,e)=\hspace{-5mm}  \sum_{i=1,e^*_i \notin E_V}^p \hspace{-5mm} \sigma^*_2 (e^*_i) w_{e^*_{i}}
\\
&\stackrel{\eqref{eqn:att_strategy_analytical}}{=} \sum_{i=1,e^*_i \notin E_V}^p \hspace{-2mm}\frac{1/w_{e^*_{i}} }{S_p} w_{e^*_{i}} 
=\sum_{i=1,e^*_i \notin E_V}^p\hspace{-2mm} \frac{1}{S_p} \stackrel{(**)}{\geq} \ \frac{p-b_1}{S_p},
\end{align*}
where (**) follows from the fact that every component $e^*_i$ belongs to a different monitoring set, so at most $b_1$ of them can be monitored by placing sensors at nodes $V$.
Thus, we can conclude that  $\frac{p-b_1}{S_p}$ is the value of the game, 
and $({\sigma}^*_{1},{\sigma}^*_{2})$ is a NE of $\Gamma$.
 \end{proof}

We now discuss P1's equilibrium strategy. 
From~\eqref{eqn:def_strategy_analytical}, we see that the probability of  P1 placing a sensor at node $v_j$ depends on the corresponding maximum criticality $w_{e^*_j}$: the higher $w_{e^*_j}$ is, the higher the probability of placing a sensor at $v_j$ is. 
This is intuitive because P1 monitors more critical components with higher probability. 
Additionally, note that P1 places sensors only on the first $p$ nodes.
If $p$$<$$n$, nodes $v_{p+1},\ldots,v_n$ are never allocated any sensor, and hence, the components from $E_{v_{p+1}},\ldots,E_{v_n}$ are never monitored. 
This is in contrast with the result from~\cite{2017arXiv170500349D}, where it was shown that P1 monitors every component with non-zero probability in any NE. 
Indeed, in our proof, we show that  the unmonitored components have criticality lower than the value of the game.  
Another interesting observation is that the set of nodes on which sensors are allocated also depends on the number of sensors P1 has at her disposal. 
Particularly, the more sensors P1 has, on the more nodes she allocates sensors, as shown in the following proposition. 

\begin{proposition}\label{lemma:set_Z}
Let $b_1$$ \in$$ \mathbb{N}$ (resp. $b_1'$$ \in$$ \mathbb{N}$) be given, and $p$ (resp. $p'$) be the largest element of $Z(b_1)$ (resp. $Z(b_1')$). 
If $b_1$$<$$ b'_1$$\leq n$, then $p$$\leq$$ p'$. 
\end{proposition} 

\begin{proof}
Note that $p$ (resp. $p'$) exists, since $b_1$$ \in$$ Z(b_1)$ (resp. $b_1'$$ \in$$ Z(b_1')$). We then have
$$
w_{e^*_p}\stackrel{\eqref{eqn:Zset}}{\geq} \frac{p-b_1}{S_p}\stackrel{(*)}{>}\frac{p-b'_1}{S_p}, 
$$
where (*) holds because $b_1$$<$$ b'_1$. 
Hence, $p$$ \in $$Z(b_1')$.
Since $p'$ is the largest element of $Z(b_1')$, $p'\geq p$ must hold. 
\end{proof}


We now discuss P2's equilibrium strategy.
Firstly, it follows from~\eqref{eqn:att_strategy_analytical} that P2 targets only the components from $E_p$. 
Thus, the unmonitored components are not necessarily targeted in equilibrium, again in contrast to~\cite{2017arXiv170500349D}.
Indeed, P2 on average gains more by attacking components from $E_p$, even though they may be monitored by P1 with a non-zero probability. 
Next, observe that the components from $E_p$ with higher criticality are targeted with lower probability.
The reason is that P1 monitors high criticality components with higher probability, which results in P2 targeting these components with a lower probability to remain undetected. 
Finally, note that the number of components P2 attacks is non-decreasing with the number of sensors P1 decides to deploy; this follows from Proposition~\ref{lemma:set_Z}.

\subsection{Overlapping Monitoring Sets and Single Sensor} \label{section:second_special_case}

To better understand if some of the conclusions from Section~\ref{section:first_special_case} can be extended to the case of overlapping monitoring sets, we discuss the case of single sensor ($b_1$$=$$1$). 
We introduce the following primal and dual linear programs that characterize the equilibrium for this game instance: 
  \begin{align*}
(\mathcal{P})\hspace{1mm}&\underset{x\geq 0 }{\text{maximize}}\hspace{1mm}\sum_{v \in \mathcal{V}} x_{v}\hspace{2mm}
\text{subject to } \sum_{\substack{v \in \mathcal{V}\\ e \notin E_{v}}} x_{v} \leq \frac{1}{w_e}, \forall e\in\mathcal{E}, \\
(\mathcal{D})\hspace{1mm}&\underset{y\geq 0 }{\text{minimize}}\hspace{1mm}\sum_{e \in \mathcal{E}} \frac{y_{e}}{w_{e}}\hspace{2mm}
\text{subject to } \sum_{ \substack{e \in \mathcal{E} \\ e \notin E_{v}}} y_{e}\hspace{-0.5mm} \geq1, \forall v \in\mathcal{V}.
 \end{align*}
These problems are reformulations of $\text{LP}_1$ and $\text{LP}_2$~\cite[Section 2]{basar1999dynamic}.
Under the reasonable assumption that P1 cannot monitor all the components using a single sensor, $(\mathcal{P})$ and $(\mathcal{D})$ are bounded. 
Moreover, thanks to strong duality, their optimal values coincide.  
Let $x^*$ be a solution of $(\mathcal{P})$, $y^*$ be a solution of $(\mathcal{D})$, and $J^*$ be the optimal value of these programs. 
Then the following strategy profile
 \begin{align} \label{eqn:lin_prog}
\bar{\sigma}^*_1(v)= \frac{x^*_{v}}{ J^*}, \hspace{5mm}\bar{\sigma}^*_2(e)=  \frac{y^*_{e}}{J^*w_{e}}, 
  \end{align}
is a NE of $\Gamma$. 
\begin{proposition}\label{theorem:solution_special_case_3}
Let $b_1$$=$$1$, and assume that $E_v$$ \neq $$\mathcal{E}$ for any $v \in \mathcal{V}$. 
The strategy profile~\eqref{eqn:lin_prog} is a NE of~$\Gamma$, and the value of the game is $L(\bar{\sigma}^*_1,\bar{\sigma}^*_2)=\frac{1}{J^*}$. 
\end{proposition}
\begin{proof}
If $E_v$$ \neq $$\mathcal{E}$ for any $v $$\in $$\mathcal{V}$, then $(\mathcal{D})$ is feasible. 
For example, $y_{e_1}$$=$$\ldots$$=$$y_{e_m}$$=$$1$ represents a feasible solution of $(\mathcal{D})$. 
Thus, $J^*$ is bounded and the strategy profile~\eqref{eqn:lin_prog} is well-defined. 
Now, for any $e $$\in$$ \mathcal{E}$, we have 
  \begin{align*}
  L (\bar{\sigma}^*_1,e)  = \sum_{v \in \mathcal{V}} \bar{\sigma}^*_1(v) l(v,e)
  \stackrel{\eqref{eqn:index_and_set_x},\eqref{eqn:lin_prog}}{=} 
 \frac{ w_e}{J^*}  \sum_{v \in \mathcal{V}, e \notin E_{v}}  x^*_{v}.
    \end{align*}
    Note that
$ w_e\sum_{v \in \mathcal{V}, e \notin E_{v}} x^*_{v}$$\leq $$1$, since $x^*$ is a solution of $(\mathcal{P})$. 
    Thus,   $L(\bar{\sigma}^*_1,e)$$\leq$$ \frac{1}{J^*}$ for any $e $$\in $$\mathcal{E}$.
        Similarly, for any $v$$ \in$$ \mathcal{V}$
  \begin{align*}
  L (v,\bar{\sigma}^*_2)  = \sum_{e \in \mathcal{E}} \bar{\sigma}^*_2(e) l(v,e)  \stackrel{\eqref{eqn:index_and_set_x},\eqref{eqn:lin_prog}}{=} 
 \sum_{ e\in \mathcal{E}, e \notin E_{v} } \frac{y^*_{e} }{  J^*} ,
    \end{align*}
where $\sum_{e\in \mathcal{E}, e \notin E_{v}} $$y^*_{e}  $$\geq  $$1$  since $y^*$ is a solution of $(\mathcal{D}$). 
  Thus,   $L (v,\bar{\sigma}^*_2)\geq \frac{1}{J^*}$ for any $v \in \mathcal{V}$.
Hence, $\frac{1}{J^*}$ is the value of the game, and $(\bar{\sigma}^*_1,\bar{\sigma}^*_2)$ is a NE. 
\end{proof}

To understand P1's equilibrium strategy, note that  $x^*_{v}$ can be viewed as a scaled probability of inspecting $v$. 
By inserting $x^*$ into the constraints of $(\mathcal{P})$, and  dividing them by $J^*$, we obtain 
$  \sum_{v \in \mathcal{V},e \notin E_{v}} \frac{x^*_{v}}{J^*} $$\leq$$ \frac{1}{ w_e}L(\bar{\sigma}^*_1,\bar{\sigma}^*_2),$$\forall $$e\in\mathcal{E}.$
One can now verify that the left side of this inequality is the probability of \textit{not} monitoring $e$. 
Thus, if $w_e $$\leq$$L(\bar{\sigma}^*_1,\bar{\sigma}^*_2)$, then P1 can leave $e$ unmonitored. 
Otherwise, P1 monitors $e$ with non--zero probability. 
Additionally, the higher $w_e$ enforces the lower probability that $e$ is left unmonitored. 
Note that all these observations are similar to the ones we made for the case discussed in Section~\ref{section:first_special_case}.

In P2's equilibrium strategy, $y^*_{e}$ can be interpreted as the scaled gain that P2 achieves by targeting $e$. 
Namely, by inserting $y^*$ into the constraints of $(\mathcal{D})$, and dividing them by $J^*$, we obtain 
$
\sum_{e\in \mathcal{E},e \notin E_{v}}$$ \frac{y^*_{e}}{J^*} $$\geq $$L(\bar{\sigma}^*_1,\bar{\sigma}^*_2),$$\forall $$v$$\in$$\mathcal{V}.
$
The left hand side of the inequality represents P2's payoff once P1 monitors $v$.
Thus, the constraints of $(\mathcal{D})$ guarantee that P2's payoff is at least $\frac{1}{J^*}$.
Next, P2's objective is to minimize $\sum_{e \in \mathcal{E}}$$ \frac{y_{e}}{w_{e}}$, so she has more incentive to increase $y_{e}$ for which the corresponding criticality $w_{e}$ is higher. 
This is consistent with the attack strategy~\eqref{eqn:att_strategy_analytical}, where P2 targeted the components from $E_p$.
Additionally, assume that the components $e_1$ and $e_2$ are associated with the same value of the scaled gain, that is,  $y_{e_1}$$=$$y_{e_2}$. 
It then follows from~\eqref{eqn:lin_prog} that the component with higher criticality has the lower probability of being targeted by P2, which is another similarity with~\eqref{eqn:att_strategy_analytical}.

Although the discussion above provides us with some game-theoretic intuition, we are unable to say more about a NE~\eqref{eqn:lin_prog} since $x^*$ and $y^*$ are unknown. 
In the next section, we introduce an $\epsilon$-NE that can give us more insights about equilibrium strategies in the general case of the game.

\section{Approximate Equilibrium Strategies} \label{section:approximate_strategies}

In this section, we show that the mixed strategies developed in~\cite{2017arXiv170500349D} can be used to obtain an $\epsilon$-NE for $\Gamma$, and discuss possible ways to improve the monitoring strategy from this $\epsilon$-NE.
We begin by introducing necessary preliminaries. 

\subsection{Preliminaries}

We first define set packings and set covers, which are two essential notions that we use subsequently.

\begin{definition}
We say that $E \in 2^\mathcal{E}$ is: (1)~A set packing, if for all $v \in \mathcal{V}$, $|E_{v} \cap E|\leq 1$; (2)~A \textit{maximum} set packing, if $|E'|\leq |E|$ holds for every other set packing $E'$. 
\end{definition}

\begin{definition}
We say that $V \in 2^\mathcal{V}$ is:
(1)~A set cover, if $E_V=\mathcal{E}$; 
(2)~A \textit{minimum} set cover if $|V|\leq |V'|$ holds for every other set cover $V'$.
\end{definition}

Set packings are of interest to P2. 
Namely, each of the components from a set packing needs to be monitored by a separate sensor.  
Thus, by randomizing the attack over a set packing, P2 can make it more challenging for P1 to detect the attack.
Similarly, set covers are of interest for P1.
In fact, if P1 is able to form a set cover using $b_1$ sensors, she can monitor all the components.
In that case, $\Gamma$ is easy to solve in pure strategies, as shown in the following proposition.

\begin{proposition} \label{thm:purestrategies}
A pure strategy profile ($V^*,e^*$) is a  NE of $\Gamma$ if and only if $V^*$ is a set cover and $|V^*|\leq b_1$. 
\end{proposition}
\begin{proof}
($\Rightarrow$) The proof is by contradiction. 
Let $(V^*,e^*)$ be a NE in which $V^*$ is not a set cover. 
Assume first that $l(V^*,e^*)$$=$$0$.
Since $V^*$ is not a set cover, P2 can attack $e$$ \notin$$ E_{V^*}$. 
Then $l(V^*,e)$$=$$w_e$$>$$l(V^*,e^*)$$=$$0$, so  $(V^*,e^*)$ cannot be a NE. 
The remaining option is $l(V^*,e^*)$$>0$.
In this case, P1 can select to play $V$, $e^* $$\in $$E_{V}$, and decrease the loss to 0. 
Thus, $(V^*,e^*)$ cannot be a NE in this case either. 

($\Leftarrow$) If $|V^*|$$\leq$$ b_1$, then $V^*$$ \in$$ \mathcal{A}_1$.
Furthermore, if $V^*$ is a set cover, then $l(V^*,e)$$=$$0$ for all $e$$ \in$$ \mathcal{A}_2$.
Thus, P1 cannot decrease the loss any further, and P2 cannot increase it, which implies $(V^*,e^*) $ is a NE.
\end{proof}

A more interesting and practically relevant situation is one in which P1 is not able to monitor all the components simultaneously due to limited sensing budget. 
Therefore, we henceforth assume that P1 cannot form a set cover using $b_1$ sensors; i.e. $b_1$$ < $$|V|$ hold for any set cover $V$$ \in$$ 2^\mathcal{V}$.

\subsection{Set Cover/Set Packing Based Strategies}


We now introduce the mixed strategies constructed using the notion of  minimum set cover and maximum set packing. 
Particularly, let  $V^*$ (resp. $E^*$) be a minimum set cover (resp. a maximum set packing), and  
$n^*$$\coloneqq$$ |V^*|$ (resp. $m^*$$\coloneqq $$|E^*|$). 
Following~\cite{2017arXiv170500349D}, we consider the mixed strategies $\sigma^\epsilon_{1}$ and $\sigma^\epsilon_{2}$ characterized by
\begin{align}
\label{eqn:def_strategy_covers}
\rho_{\sigma^\epsilon_{1}}(v)&=
\begin{cases} 
\frac{b_1}{n^*},\hspace{2mm}v\in V^*,  \\
\hspace{2.2mm}0,\hspace{2mm}v\notin V^*,
\end{cases}\\
\label{eqn:att_strategy_packings}
\sigma^\epsilon_{2}(e)&=
\begin{cases}\frac{1}{m^*},\hspace{2mm}e\in E^*,  \\
\hspace{2.5mm} 0,\hspace{2.2mm}e\notin E^*. 
\end{cases}
\end{align}
In other words, P1 places sensors only on nodes from $V^*$ with probability $\frac{b_1}{n^*}$. 
Since $V^*$ is a set cover, it follows that every component is monitored with probability at least $\frac{b_1}{n^*}$.
The strategy of P2 is to attack the components from $E^*$ with probability $\frac{1}{m^*}$. 
The proof of existence of a strategy profile $(\sigma^\epsilon_{1},\sigma^\epsilon_{2})$ satisfying~\eqref{eqn:def_strategy_covers}--\eqref{eqn:att_strategy_packings} is by construction, and can be found in \cite[Lemma~1]{2017arXiv170500349D}. 
Let 
$w_{\min}$$\coloneqq$$\min_{e \in \mathcal{E}}$$w_e$, $w_{\max}$$ \coloneqq$$\max_{e \in \mathcal{E}}$$ w_e$, and $\Delta_w $$\coloneqq$$ w_{\max}$$-$$w_{\min}.$
The following theorem establishes that $(\sigma^\epsilon_{1},\sigma^\epsilon_{2})$ is an $\epsilon$--NE, and gives the worst case values for~$\epsilon$ and P1's loss.

\begin{theorem} \label{thm:mix_strategies_diff_indexes} 
Any strategy profile that satisfies~\eqref{eqn:def_strategy_covers}--\eqref{eqn:att_strategy_packings} is an $\epsilon$-NE of $\Gamma$,  where
$$\epsilon=\underbrace{b_1 w_{\min} \frac{n^*-\max\{b_1,m^*\}}{n^*\max\{b_1,m^*\}}}_{=\epsilon_1}+\underbrace{\Delta_w\frac{n^*-b_1}{n^*}}_{=\epsilon_2}.$$
Furthermore, for any $\sigma_2 $$\in $$\Delta_2$, we have 
\begin{equation} \label{eqn:worst_case_loss_set_cover}
L(\sigma^{\epsilon}_1,\sigma_2) \leq  w_{\max} \frac{n^*-b_1}{n^*}.
\end{equation}
\end{theorem}

\begin{proof}
We first derive an upper bound on P1's expected loss if she plays ${\sigma}^{\epsilon}_{1}$. 
Let $e$ be an arbitrary component, and $ \mathcal{A}_{1}' $$=$$ \{V $$\in$$ \mathcal{A}_1$$|$$ l(V,e)$$=$$w_e$$ \}$ be the set of sensor placements in which $e$ is not monitored. 
The expected loss $L({\sigma}^{\epsilon}_{1},e)$ is then
\begin{align*}
L({\sigma}^{\epsilon}_{1},e)&=  \sum_{V \in \mathcal{A}_1}\sigma^{\epsilon}_1(V) l(V,e)=w_e  \sum_{V \in \mathcal{A}_{1}'}  \sigma^{\epsilon}_1(V).
\end{align*}
Note that $\sum_{V \in \mathcal{A}_{1}'}\sigma^{\epsilon}_1(V)$ represents the probability that  $e$ is not monitored. 
This probability is at most $1-\frac{b_1}{n^*}$, since P1 inspects every element of a set cover with probability $\frac{b_1}{n^*}$.
Moreover, $w_{e}\leq w_{\max}$.
It then follows that
\begin{equation} \label{eqn:lb_equal_prob}
\begin{aligned}
L(\sigma^{\epsilon}_1,e) \leq  w_{\max} \frac{n^*-b_1}{n^*}=\bar{L}, 
\end{aligned}
\end{equation}
which confirms~\eqref{eqn:worst_case_loss_set_cover}.
We now derive a lower bound on the expected payoff of P2 if she plays $\sigma^{\epsilon}_2$.  
Let $V$ be an arbitrary element of $\mathcal{A}_{1}$,
and $ E' $$= $$ \{e  $$\in $$ E^*| l(V,e) $$= $$w_e\}$ be the set of components from $E^*$ that are not monitored from $V$. 
Then
\begin{align*}
 L(V,\sigma^{\epsilon}_2) &=  \sum_{e \in \mathcal{E}}\sigma^{\epsilon}_2(e) l(V,e) \stackrel{\eqref{eqn:att_strategy_packings}}{=} \frac{1}{m^*} \sum_{e \in E' } w_e \\ &  \geq 
\frac{1}{m^*}\sum_{e \in E' } w_{\min} =\frac{ |E' |}{m^*} w_{\min}.
\end{align*}
Since $E^*$ is a maximum set packing and $|V|\leq b_1$, at most $b_1$ components can be monitored by positioning $V$.
Therefore, $|E'|\geq \max\{0,m^* -b_1\}$,
and we conclude
\begin{equation} \label{eqn:ub_equal_prob}
L(V,\sigma^{\epsilon}_2)
\geq  w_{\min}\frac{ \max{ \{0,m^* -b_1\} }}{m^*}=\underline{L}.
\end{equation}
From~\eqref{eqn:lb_equal_prob} and~\eqref{eqn:ub_equal_prob}, it follows that 
$\underline{L} \leq L(\sigma^{\epsilon}_1,\sigma^{\epsilon}_2)\leq \bar{L}$. 
Thus, $(\sigma^{\epsilon}_1,\sigma^{\epsilon}_2)$ is an $\epsilon$-NE, where 
\begin{align*}
\epsilon&= \bar{L}-\underline{L}=w_{\max} \frac{n^*-b_1}{n^*} - w_{\min}\frac{ \max{ \{0,m^* -b_1\} }}{m^*} \\
&= (w_{\min}+\Delta_w) \frac{n^*-b_1}{n^*} - w_{\min}\frac{ \max{ \{0,m^* -b_1\} }}{m^*}\\
&= b_1 w_{\min} \frac{n^*-\max\{b_1,m^*\}}{n^*\max\{b_1,m^*\}}+\Delta_w\frac{n^*-b_1}{n^*}. \text{\hspace{13mm}}
\end{align*}
This concludes the proof. \end{proof}

From Theorem~\ref{thm:mix_strategies_diff_indexes}, we can draw the following conclusions. 
If all the components have equal criticality level, then $\Delta_w$$=$$0$ and $\epsilon_2$$=$$0$. 
In that case, $\epsilon_1$$=$$0$ if $n^*$$=$$m^*$, and $(\sigma^\epsilon_{1},\sigma^\epsilon_{2})$ is an exact NE. 
Although $n^*$$=$$m^*$ may look as a restrictive condition, it turns out that $n^*$ and $m^*$ are often equal or close to each other in practice~\cite{2017arXiv170500349D}.
Also note that the strategy profile constructed using~\eqref{eqn:def_strategy_covers}--\eqref{eqn:att_strategy_packings} differs from the equilibrium profile developed in Section~\ref{section:first_special_case} in two aspects: 
(i)~Since $V^*$ is a set cover, every component is monitored with non-zero probability;
(ii)~The set of nodes where sensors are placed (resp. the set of  attacked components) does not change with $b_1$, that is, it is always $V^*$ (resp. $E^*$).

However, if $\Delta_w$ is large, $\epsilon$ can be large even if $n^*$$=$$m^*$. 
The strategies $\sigma^\epsilon_{1}$ and $\sigma^\epsilon_{2}$ may fail in this case because they assume every component to be equally critical. 
For instance, consider the case from Fig.~\ref{figure:example_2}.
We have $V^*$$=$$\{v_1,v_2\}$, $E^*$$=$$\{e_1,e_3\}$, the criticality of blue (resp. red) components is $w_{\min}$ (resp. $w_{\max}$), and $b_1$$=$$1$. 
From Fig.~\ref{figure:example_2}~a), we see that P1 monitors $e_1$ and $e_3$ with equal probability, although they have different criticality levels. 
Thus, the best response of P2 is to target $e_3$, which results in the worst case loss of P1.  
Similarly, as seen in Fig.~\ref{figure:example_2} b), P2 targets the components $e_1$ and $e_3$ with equal probability.
The best response of P1 is then to monitor $e_3$, leaving P2 with the lowest payoff.


     \begin{figure}[t]   
    \centering
  \includegraphics[width=75mm]{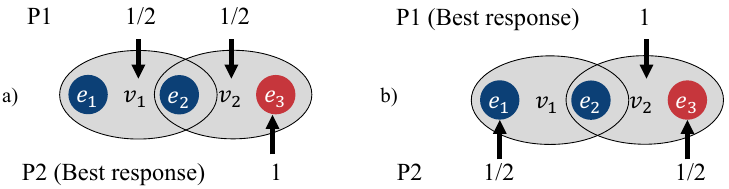}
  \caption{The figure illustrates why the strategies $\sigma^\epsilon_{1}$ and $\sigma^\epsilon_{2}$ may fail.  The criticality of red (resp. blue) components is $w_{\max}$ (resp. $w_{\min}$).}
  \label{figure:example_2}
\end{figure}


Nevertheless, the set cover strategy $\sigma_1^\epsilon$ has several favorable properties.  
Firstly, we note that by playing $\sigma_1^\epsilon$, P1 cannot lose more than~\eqref{eqn:worst_case_loss_set_cover}. 
Thus, if $b_1$ is close to $n^*$, the worst case loss~\eqref{eqn:worst_case_loss_set_cover} and $\epsilon$ approach 0, and $\sigma_1^\epsilon$ represents a good approximation for equilibrium monitoring strategy.      
If $b_1$$=$$n^*$, both the worst case loss~\eqref{eqn:worst_case_loss_set_cover} and $\epsilon$ are 0, and $\sigma^\epsilon_1$ becomes a pure equilibrium strategy from Proposition~\ref{thm:purestrategies}.
Secondly, this strategy is easy to construct. 
Namely, once $V^*$ is known, one can straightforwardly find $\sigma_1^\epsilon$ that satisfies~\eqref{eqn:def_strategy_covers} (see~\cite[Lemma~1]{2017arXiv170500349D}).
Although calculating $V^*$ is NP--hard problem, modern integer linear program solvers can obtain a solution of this problem for relatively large values of $n$, and greedy heuristics can be used for finding an approximations of $V^*$ with performance guarantees~\cite{chvatal1979greedy}. 
Finally, $\sigma_1^\epsilon$ can be further improved in several ways, as discussed next.

\subsection{Improving the Set Cover Monitoring Strategy}


\subsubsection{Increasing $b_1$} As we already mentioned, both the worst case loss~\eqref{eqn:worst_case_loss_set_cover} and $\epsilon$ approach 0 when $b_1$ approaches $n^*$. 
Thus, an obvious way to improve $\sigma_1^\epsilon$ is by increasing $b_1$.

\subsubsection{Focusing on highest criticality components} 
Assume a situation where a group of components $\bar{\mathcal{E}}$ have criticality $w_{\max}$ that is much larger compared to the criticality of the remaining components. 
In Section~\ref{section:exact}, we showed that depending on $b_1$ and the components' criticality, P1 (resp. P2) may focus on monitoring (resp. attacking) the components with the highest criticality, while neglecting the others.
Let $\bar{w}_{\max}$ be the largest criticality among the components $\mathcal{E}$$\setminus$$ \bar{\mathcal{E}}$. 
We show that if $\bar{\Delta}_w$$:=$$w_{\max}$$-$$\bar{w}_{\max}$$ \geq$$ w_{\max} \frac{ b_1}{\bar{n}^*}$ , a small modification of the strategies~\eqref{eqn:def_strategy_covers}--\eqref{eqn:att_strategy_packings} can give us a potentially improved $\epsilon$-NE.  
Particularly, let $\bar{V}^*$ (resp. $\bar{E}^*$) be a minimum set cover for $\bar{\mathcal{E}}$ (resp. maximum set packing of  $\bar{\mathcal{E}}$), $|\bar{V}^*|$$ \coloneqq$$\bar{n}^*$, $|\bar{E}^*| $$\coloneqq$$ \bar{m}^*$, and $(\bar{\sigma}^\epsilon_{1},\bar{\sigma}^\epsilon_{2})$ be a strategy profile that satisfies 
\begin{align}
\label{eqn:def_strategy_covers_bin}
\bar{\rho}_{\sigma^\epsilon_{1}}(v)&=
\begin{cases} 
\frac{b_1}{\bar{n}^*},\hspace{2mm}v\in \bar{V}^*,  \\
\hspace{2mm}0,\hspace{2mm}v\notin \bar{V}^*,
\end{cases}\\
\label{eqn:att_strategy_packings_bin}
\bar{\sigma}^\epsilon_{2}(e)&=
\begin{cases}\frac{1}{\bar{m}^*},\hspace{2mm}e\in \bar{E}^*,  \\
\hspace{3mm} 0,\hspace{2.2mm}e\notin \bar{E}^*. 
\end{cases}
\end{align}
In other words, P1 (resp. P2) focuses on monitoring (resp. targeting) the components $\bar{\mathcal{E}}$ using the strategy $\bar{\sigma}^\epsilon_{1}$ (resp. $\bar{\sigma}^\epsilon_{2}$). 
The proof that $(\bar{\sigma}^\epsilon_{1},\bar{\sigma}^\epsilon_{2})$ exists is the same as for $(\sigma^\epsilon_{1},\sigma^\epsilon_{2})$. 
The following then holds. 

\begin{proposition} \label{thm:binary_weights} 
If $\bar{\Delta}_w$$\geq $$  w_{\max}\frac{ b_1}{\bar{n}^*}$, then any strategy profile that satisfies~\eqref{eqn:def_strategy_covers_bin}--\eqref{eqn:att_strategy_packings_bin} is an $\bar{\epsilon}$-NE of $\Gamma$,  where
$$\bar{\epsilon}=b_1 w_{\max} \frac{\bar{n}^*-\max\{b_1,\bar{m}^*\}}{\bar{n}^*\max\{b_1,\bar{m}^*\}}.$$
Furthermore, for any $\sigma_2 $$\in $$\Delta_2$, we have 
\begin{equation} \label{eqn:worst_case_loss_set_cover_2}
L(\bar{\sigma}^\epsilon_{1},\sigma_2 ) \leq  w_{\max} \frac{\bar{n}^*-b_1}{\bar{n}^*}.
\end{equation}
\end{proposition}
\begin{proof}
Assume P1 plays according to~\eqref{eqn:def_strategy_covers_bin}. 
If P2 attacks $e $$\in  $$\bar{\mathcal{E}}$, we can show using the same reasoning as in the proof  of Theorem~\ref{thm:mix_strategies_diff_indexes} that $L(\bar{\sigma}^\epsilon_{1},e) $$ \leq $$  w_{\max} \frac{\bar{n}^*-b_1}{\bar{n}^*} $$= $$\bar{L}. $
If P2 attacks $e$$\in $$ \mathcal{E} $$\setminus $$ \bar{\mathcal{E}}$, we have
$$L(\bar{\sigma}^\epsilon_{1},e)
\stackrel{(*)}{\leq} \bar{w}_{\max} \stackrel{(**)}{\leq}  w_{\max} \frac{\bar{n}^*-b_1}{\bar{n}^*}=\bar{L}, $$
where (*) follows from the fact that the largest loss occurs when $e$ is unmonitored and has criticality $\bar{w}_{\max}$, and (**) from $\bar{\Delta}_w$$\geq $$w_{\max}$$ \frac{ b_1}{\bar{n}^*}$.  
Thus, P1 looses at most $\bar{L}$ by playing according to $\sigma^\epsilon_{1}$.
If P2 plays according to~\eqref{eqn:att_strategy_packings_bin}, we obtain
\begin{equation*}
L(V,\bar{\sigma}^\epsilon_{2})\geq  w_{\max}\frac{ \max{ \{0,\bar{m}^* -b_1\} }}{\bar{m}^*}=\underline{L},
\end{equation*}
by following the same steps as in the proof of Theorem~\ref{thm:mix_strategies_diff_indexes}. 
Thus, $(\bar{\sigma}^\epsilon_{1},\bar{\sigma}^\epsilon_{2})$ is an $\bar{\epsilon}$--NE with $\bar{\epsilon}=\bar{L}-\underline{L}$. 
\end{proof}

Proposition~\ref{thm:binary_weights} has two consequences. 
Firstly, since $\bar{n}^*$$\leq $$n^*$, the worst case loss~\eqref{eqn:worst_case_loss_set_cover_2} achieved with strategy $\bar{\sigma}_1^\epsilon$ cannot be larger than the one given by~\eqref{eqn:worst_case_loss_set_cover}.  
Secondly, if $\bar{n}^*$$=$$\bar{m}^*$, we have that any strategy profile that satisfies ~\eqref{eqn:def_strategy_covers_bin}--\eqref{eqn:att_strategy_packings_bin} is a NE, so
$\bar{\sigma}_1^\epsilon$ is an equilibrium monitoring strategy. 

\subsubsection{Numerical approach} 
We now briefly explain how CGP~\cite{desrosiers2005primer} can be used for improving the set cover monitoring strategy $\sigma_1^\epsilon$. 
We refer the interested reader to the Appendix for more details.
We begin by rewriting $\text{LP}_1$ in the form 
\begin{equation}\label{eqn:cg_problem}
\underset{\sigma_1 \geq0,z_1 \geq0}{\text{minimize}} \hspace{2mm} z_1 \hspace{3mm} \text{subject to }A \sigma_1+\textbf{1} z_1 \geq 0, \hspace{1mm}\textbf{1}^T\sigma_1=1,
 \end{equation}
 where $A$ is a matrix representation of $\Gamma$.  
Note that every element of $\sigma_1$ corresponds to a possible pure strategy from $\mathcal{A}_1$.
Since the number of pure strategies grows quickly with $b_1$, we cannot directly solve~\eqref{eqn:cg_problem} due to the size of decision vector.
However, the number of inequality constrains is always $m$, which allow us to use CGP to solve~\eqref{eqn:cg_problem}.  

The first step of CGP is to solve the master problem, which is obtained from~\eqref{eqn:cg_problem} by considering only a subset $\tilde{\mathcal{A}}_1$ of pure strategies.
Hence, to form the master problem, we only generate columns of $A$ that correspond to variables $\tilde{\mathcal{A}}_1$, which explains the name of the procedure. 
In our case, we initialize $\tilde{\mathcal{A}}_1$ with those pure strategies that are played with non-zero probability once P1 employs the set cover monitoring strategy $\sigma^\epsilon_1$ (see~\cite[Lemma~1]{2017arXiv170500349D} for construction of these strategies).  
Once a solution $(\tilde{z}_1^*,\tilde{\sigma}_1^*)$ of the master problem is calculated,
one solves the sub-problem 
\begin{equation}\label{eqn:cg_SP_problem}
\text{maximize}_{V \in \mathcal{A}_1} \hspace{2mm} (\rho^*)^T a_V + \pi^*,
\end{equation}
where $(\rho^*,\pi^*)$ is a dual solution of the master problem and $a_V$ is the column of $A$ that corresponds to a pure strategy $V$. 
If the optimal value of~\eqref{eqn:cg_SP_problem} is negative, $\tilde{z}_1^*$ can be decreased.  
We then add a solution of~\eqref{eqn:cg_SP_problem} to $\tilde{\mathcal{A}}_1$, and proceed to the next iteration. 
Otherwise, $\tilde{z}_1^*$ (resp. $\tilde{\sigma}_1^*$) is the optimal value of the game (resp. an equilibrium monitoring strategy), and we stop the procedure.  

The key point of CGP is to be able to solve~\eqref{eqn:cg_SP_problem} efficiently, which is not necessarily the case for every linear program.
However, in case of $\text{LP}_1$, $A$ is determined based on the loss function $l$ and has a structure that allow us to obtain a solution and the optimal value of~\eqref{eqn:cg_SP_problem} by solving a binary linear program.
Additionally, this program has $n$$+$$m$ binary decision variables and $m$$+$$1$ constrains for any $b_1$, so it can be solved efficiently for relatively large values of $n$ and $m$ using modern solvers.
This allow us to use CGP to find or approximate an equilibrium monitoring strategy for the networks of relatively large size, as shown in the next section.


\section{Numerical Study}\label{section:simulations}

We now test CGP on benchmarks of large scale water networks ky4 and ky8~\cite{jolly2013research}.
These networks can be modeled with a directed graph.
The vertices of the graph model pumps, junctions, and water tanks.
The edges model pipes, and the edge direction is adopted to be in the direction of the water flow. 
We consider attacks where P2 injects contaminants in a water network, while P1 allocates sensors to detect contaminants. 
In this case, $\mathcal{E}$ are the locations where contaminants can be injected, and $\mathcal{V}$ are the locations where sensors can be placed.
We adopt $\mathcal{E}$ and $\mathcal{V}$ to be the vertices of the water network graph. 
The monitoring sets were formed as follows: if a water flow from contamination source $e$ passes through $v$, then $e$ belongs to $E_v$~\cite{de2019optimal}.
 Criticality $w_{e}$ in this case can characterize the normalized population affected by contaminants injected in $e$~\cite{berry2005sensor}. 
For simplicity, we generated $w_e$ randomly.
We remark that $n$$=$$m$$=$$964$ (resp. $n$$=$$m$$=$$1332$) for ky4 (resp. ky8) network. 

We first measured how much time does it take to construct the set cover monitoring strategy $\sigma_1^\epsilon$, and to further improve it to an equilibrium monitoring strategy using CGP. 
We considered ky4 and ky8 networks, and varied $b_1$.
The results are shown in Fig.~\ref{figure:sim_1}. 
Notice that the longest running time was 1180 seconds, which demonstrates that CGP may allow us to improve $\sigma_1^\epsilon$ to an equilibrium monitoring strategy for the networks of relatively large size. 
However, we also see that the running time rapidly grows with $b_1$ and the network size.
This indicates that this way of calculating an equilibrium monitoring strategy may become inefficient if the network size exceeds several thousand nodes.

     \begin{figure}[t]   
    \centering
  \includegraphics[width=80mm]{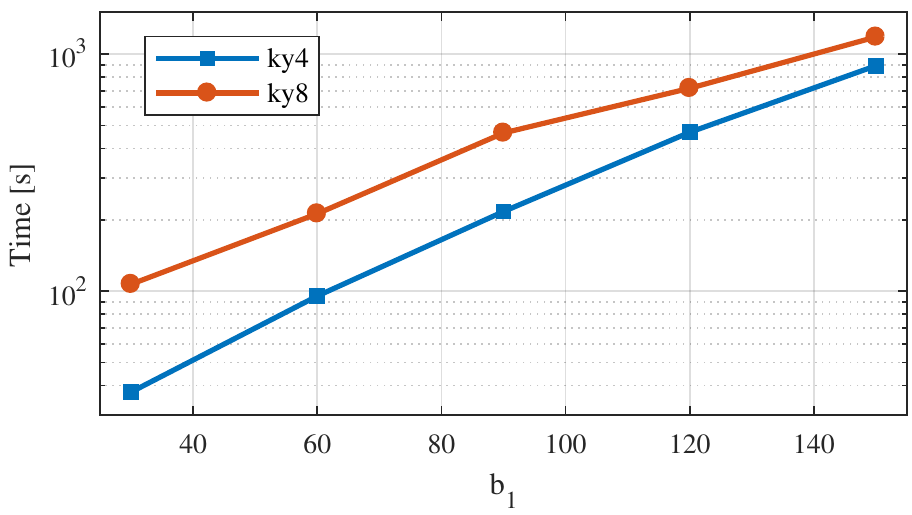}
  \caption{ Time needed to calculate an equilibrium monitoring strategy using CGP for different values of $b_1$. }
  \label{figure:sim_1}
\end{figure}

    \begin{figure}[t]   
    \centering
  \includegraphics[width=80mm]{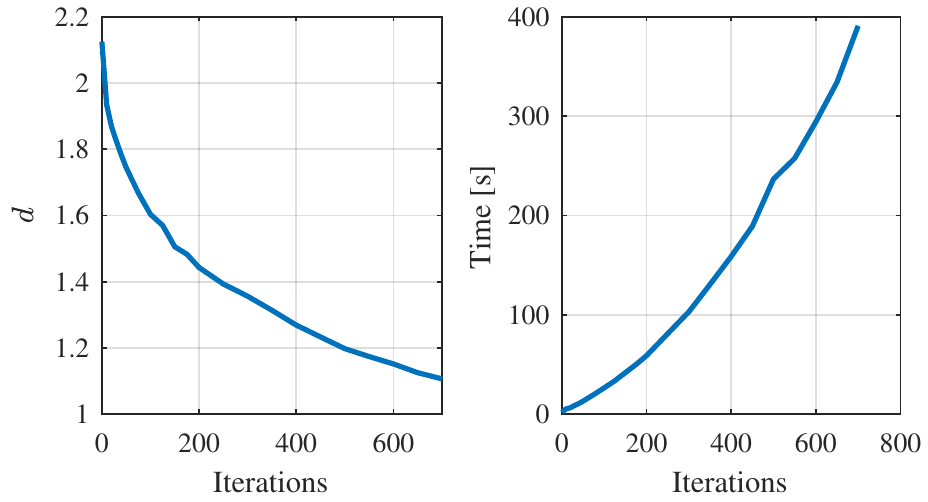}
  \caption{ Improving the set cover monitoring strategy $\sigma_1^\epsilon$ by running a limited number of CGP iterations. }
  \label{figure:sim_2}
\end{figure}

Therefore, we also explored how much can we improve $\sigma_1^\epsilon$ by running only a limited number of iterations of CGP. 
We considered ky8 network, and adopted $b_1$$=$$150$. 
As the performance metric, we used the ratio $d(i)$$:=$$\bar{L}(i)/L(\sigma^*_1,\sigma^*_2),$ where $\bar{L}(i)$ is the optimal value of the master program after $i$ iterations. 
The value $\bar{L}(i)$ upper bounds the value of the game, and represents the worst case loss
of P1 if she uses a monitoring strategy obtained by running $i$ iterations of CGP.
Hence, if $d(i)$$=$$1$, then $\bar{L}(i)$$=$$L(\sigma^*_1,\sigma^*_2)$, and CGP recovers an equilibrium monitoring strategy after $i$ iterations. 

The plot of $d$ and the execution time with respect to the number of iterations is shown in Fig.~\ref{figure:sim_2}.
Same as in the previous experiment, the execution time includes the time to construct the set cover monitoring strategy $\sigma_1^\epsilon$.
Although initially $d(0)$$\approx$$2$, $d$ reaches the value 1.11 after 700 iterations. 
We also indicate that the running time to achieve this improvement was 391 seconds, which is approximately 3 times shorter compared to the time to obtain an equilibrium monitoring strategy for $b_1$$=$$150$. 
This indicates that even if CGP may not be used to improve $\sigma_1^\epsilon$ to an equilibrium monitoring strategy, we can still significantly improve this strategy by running a limited number of CGP iterations.

\section{Conclusion}\label{section:conclusion}

This paper investigated a network monitoring game, with the purpose of developing monitoring strategies. 
The operator's (resp. attacker's) goal was to deploy sensors (resp. attack a component) to minimize (resp. maximize) the loss function defined through the component criticality. 
Our analysis revealed how criticality levels impact a NE, and outlined some fundamental differences compared to the related game~\cite{2017arXiv170500349D}. 
Particularly, the operator can leave some of the noncritical components unmonitored based on their criticality and available budget, while the attacker does not necessarily need to attack these components.
Next, we proved that previously known strategies~\cite{2017arXiv170500349D} can be used to obtain an $\epsilon$--NE, and showed how $\epsilon$ depends on component criticality.
Finally, we discussed how to improve the monitoring strategy from this $\epsilon$-NE.
It was shown that if a group of the components have criticality level sufficiently larger then the others, the strategy can be improved by a simple modification.
We also demonstrated that the strategy can be improved numerically using the column generation procedure.

The future work will go into two directions. 
Firstly, we plan to find the way to characterize and analyze properties of a NE in the general case of the game.  
Secondly, we intend to generalize the game model by relaxing some of the modeling assumptions.
For instance, to allow the attacker to target several components simultaneously,  and to remove the assumption that deployed sensors are perfectly secured.

\section*{Appendix: Column Generation Procedure} \label{appedix:CG}
CGP can be used to solve linear programs with a large number of decision variables and a relatively small number of constraints~\cite{desrosiers2005primer}, such as $\text{LP}_1$. 
The first step of CGP is to solve the master problem of $\text{LP}_1$, which can be formulated~as
\begin{equation}\label{eqn:MP_LP1}
  \begin{aligned}
&\underset{\tilde{z}_1\geq0,\tilde{\sigma}_1 \geq0}{\text{minimize}} \hspace{-2mm}&&\tilde{z}_1\\
&\text{subject to}&& \sum_{V \in \tilde{\mathcal{A}}_1}a_V  \tilde{\sigma}_1(V) + \textbf{1} \tilde{z}_1 \geq 0, \\
& && \sum_{V \in \tilde{\mathcal{A}}_1}\tilde{\sigma}_1(V)=1, 
 \end{aligned}
\end{equation}
where $ a_V$$ \in$$ \mathbb{R}^m$ is given by 
\begin{equation}\label{eqn:aV}
  a_V (i)=
\begin{cases}-w_{e_i},\hspace{2mm}e_i \notin E_V,  \\
\hspace{5.8mm} 0,\hspace{2.4mm}e_i \in E_V. 
\end{cases}
\end{equation}
The only difference between~\eqref{eqn:MP_LP1} and $\text{LP}_1$ is that we consider only a subset of pure actions $\tilde{\mathcal{A}}_1$ instead of the whole set $\mathcal{A}_1$. 
As mentioned before, we initialize $\tilde{\mathcal{A}}_1$ with pure strategies that are played with non-zero probability once P1 employs the set cover monitoring strategy $\sigma^\epsilon_1$ (see~\cite[Lemma~1]{2017arXiv170500349D} for construction of these strategies).  

Let ($\tilde{z}_1^*$,$\tilde{\sigma}^*_1$) be a solution of~\eqref{eqn:MP_LP1}. 
The next step is to check if $\tilde{z}_1^*$ can be further decreased, which can be done by solving the following subproblem
  \begin{equation}\label{eqn:reduced_cost}
 \tilde{c}:= \text{minimize}_{V \in \mathcal{A}_1} \hspace{2mm} -\sum_{i=1}^m \rho^*_{i} a_V(i) -\pi^*,
 \end{equation}
 where $\rho^*$$ \in$$ \mathbb{R}^m$ (resp.  $\pi^* $$\in$$ \mathbb{R}$) is an optimal dual solution of~\eqref{eqn:MP_LP1} that corresponds to the inequality constraints (resp. equality constraint).
If $\tilde{c}$$<$$0$, $\tilde{z}_1^*$ can be further decreased.
We then add a solution of~\eqref{eqn:reduced_cost} to $\tilde{\mathcal{A}}_1$, and repeat the procedure with the new set $\tilde{\mathcal{A}}_1$. 
Yet, if $\tilde{c}$$\geq $$0$, $\tilde{z}_1^*$ is the optimal value of $\text{LP}_1$, and $\tilde{\sigma}^*_1$ is an equilibrium monitoring strategy. 

However, the crucial point of CGP is to find an efficient way to solve~\eqref{eqn:reduced_cost}.  
Namely, due to the large cardinality of $\mathcal{A}_1$, it is not tractable to simply go through all the columns $a_V$ and pick the optimal one. 
In our case, we can avoid this by solving the following binary linear program to obtain a solution and the optimal value of~\eqref{eqn:reduced_cost} 
\begin{equation}\label{eqn:subproblem}
\begin{aligned}
&\underset{x \in \{0,1\}^{n}, y  \in \{0,1\}^{m}}{\text{minimize}} \hspace{-2mm}&& \sum_{e_i \in \mathcal{E}} \rho^*_{i} w_{e_i} y_{e_i}-\pi^*\\
&\hspace{5mm}\text{subject to}&&  \sum_{ \substack{v \in \mathcal{V}\\ e\in E_{v}}} \hspace{-1mm}x_{ v} \hspace{-1mm}\geq \hspace{-1mm}1 - y_{e} , \forall e \in\mathcal{E},\\ 
& &&  \sum_{ v \in \mathcal{V}} x_{v} \leq b_1. 
 \end{aligned}
  \end{equation}
Note that this program has $n$$+$$m$ binary variables and $m$$+$$1$ constraints regardless of $b_1$. 
Therefore, modern day integer linear program solvers can obtain a solution and the optimal value of~\eqref{eqn:MP_LP1} for relatively large values of $n$ and $m$. 
We conclude by showing how to obtain a solution and the optimal value of~\eqref{eqn:reduced_cost} by solving~\eqref{eqn:subproblem}. 

\begin{lemma}
Let $\tilde{c}$ (resp. $\tilde{x},\tilde{y}$) be the optimal value (resp. a solution) of~\eqref{eqn:subproblem}.  
Let $\tilde{V}$ be formed as follows: if $\tilde{x}_{v}$$=$$0$ (resp. $\tilde{x}_{v}$$=$$1$), then $v $$\notin $$\tilde{V}$ ($v $$\in $$\tilde{V}$). 
Then $\tilde{c}$ (resp. $\tilde{V}$) is the optimal value (resp. a solution) of~\eqref{eqn:reduced_cost}. 
\end{lemma}
\begin{proof}Firstly, note that $|\tilde{V}|\leq b_1$ since $\tilde{x}$ has to satisfy the second constraint of~\eqref{eqn:subproblem}. 
Thus, $\tilde{V}$ is a feasible point of~\eqref{eqn:reduced_cost}.   
We now show that $\tilde{V}$ is a solution of~\eqref{eqn:reduced_cost}, and that  the optimal values of~\eqref{eqn:reduced_cost} and~\eqref{eqn:subproblem} coincide. 

Note that $\rho^*$$\geq$$ 0$ as a dual solution of~\eqref{eqn:MP_LP1}, $w_{e}$$>$$0$, and the objective of~\eqref{eqn:subproblem} reduces to minimizing $\sum_{e_i \in \mathcal{E}} \rho^*_{i} w_{e_i} y_{e_i}$. 
Thus, for fixed $\tilde{x}$, the best way to minimize the objective is to set as many elements of $y$ to 0. 
Yet, an element $y_{e_i}$ can be set to zero only if $\sum_{ v \in \mathcal{V},e_i\in E_{v}}$$ x_{v}$$\geq$$1$.
This happens once $e_i $$\in$$ E_{\tilde{V}}$. 
Otherwise, $y_{e_i}$$=$$1$ has to hold in order for a constraint to be satisfied.  
Hence, for a fixed $\tilde{x}$, the lowest objective value that can be achieved over all feasible $y$ is
\begin{equation}\label{eqn:appdendex_statement1} 
\tilde{c}= \sum_{e_i \in \mathcal{E},e_i \notin E_{\tilde{V}}} \rho^*_{i} w_{e_i} -\pi^*. 
\end{equation}
On the other hand, the value of the objective function from~\eqref{eqn:reduced_cost} for $\tilde{V}$ is given by 
\begin{equation}\label{eqn:appdendex_statement2} 
- \sum_{e_i \in \mathcal{E}} \rho^*_{i} a_{\tilde{V}}(i)-\pi^*\stackrel{\eqref{eqn:aV}}{=}\sum_{e_i \in \mathcal{E},e_i \notin E_{\tilde{V}}} \rho^*_{i} w_{e_i}-\pi^*\stackrel{\eqref{eqn:appdendex_statement1} }{=}\tilde{c}.
\end{equation}
From~\eqref{eqn:appdendex_statement2}, it follows that the optimal value of~\eqref{eqn:reduced_cost} is at least $\tilde{c}$. 
We now finalize the proof by showing that the optimal value of~\eqref{eqn:reduced_cost} cannot be lower than $\tilde{c}$ using contradiction.
Let $V'$ be a solution of~\eqref{eqn:reduced_cost}, and assume $c'$$ <$$ \tilde{c}$.  
Let $x'$ be constructed as follows: $x'_{v }$$=$$0$ (resp. $x'_{v }$$=$$1$) if $v $$\notin $$V'$ (resp. $v $$\in $$V'$).
Since $|V'|$$\leq$$ b_1$, $x'$ satisfies the constraints of~\eqref{eqn:subproblem}.
For this $x'$, let $y'_{e_i }$$=$$0$ (resp. $y'_{e_i }$$=$$1$) if $e_i $$\in $$E_{V'}$ ($e_i$$\notin$$E_{V'}$).
This $y'$ also satisfies the constraints of~\eqref{eqn:subproblem}, and we have  
\begin{align*}
 \sum_{e_i\in \mathcal{E} } \rho^*_{i} w_{e_i}y'_{e_i}-\pi^*&=\sum_{e_i\in \mathcal{E},e_i \notin E_{V'}} \rho^*_{i} w_{e_i}-\pi^* \\
&\stackrel{\eqref{eqn:aV}}{=}\sum_{e_i\in \mathcal{E}} \rho^*_{i} a_{V'}(i)-\pi^*=c'.
\end{align*}
This contradicts the assumption that the optimal value of~\eqref{eqn:subproblem} equals $\tilde{c}$, since $c' $$<$$ \tilde{c}$.
Thus, $V'$ cannot exist, and $\tilde{c}$ (resp. $\tilde{V}$) is the optimal value (resp. a solution) of~\eqref{eqn:subproblem}. 
\end{proof}

\balance

\bibliographystyle{IEEEtran}
\bibliography{CDC_2019_BIB}

\end{document}